\documentclass[12pt]{iopart}

\usepackage{iopams}
\usepackage{graphics}
\usepackage{bm}
\usepackage{booktabs}
\usepackage{graphicx}
\usepackage{amsthm}
\everymath{\displaystyle}

\begin{document}
\newcommand{\rmnum}[1]{\romannumeral #1}

\newtheorem{theorem}{Theorem}[section]
\newtheorem{lemma}[theorem]{Lemma}

\def\bm{\boldsymbol}
\def\mb{\boldsymbol}
\def\bsigma{\bm{\sigma}}
\def\btau{\bm{\tau}}
\def\half{\frac 1 2}
\let\dsp=\displaystyle
\def\cal{\mathcal}
\title[Accurate Evaluations of Strain and Stress]
{Accurate Evaluations of Strain and Stress in Atomistic Simulations of Crystalline Solids}

\author{Jerry Zhijian Yang$^{1}\footnote{Author to whom any correspondence should be addressed.}
$, Xiaojie Wu$^2$, Xiantao Li$^2$}

\address{
$^1$School of Mathematics and Statistics, Wuhan University, Wuhan, Hubei, 430072, China\\
$^2$Department of Mathematics, the Pennsylvania State University, University Park, PA 16802, USA}
\ead{zjyang.math@whu.edu.cn}
\begin{abstract}
In this paper, we study the accuracy of Irving-Kirkwood type of formulas for the approximation of continuum quantities from atomistic simulations.  Such formulas are derived by expressing the displacement, deformation gradient and stress in terms of certain kernel functions. We propose two criteria for
choosing the kernel functions to significantly improve the sampling accuracy. We present a simple procedure to construct kernel functions that meet these criteria. Further, numerical tests on homogeneous and non-homogeneous systems provide validations for our analysis.
\end{abstract}

\maketitle

\section{Introduction}
\label{introduction}

Molecular models have undoubtedly been one of the most important methods for material simulations.
By following particle trajectories, molecular simulations generate a large amount of data that need to be
further processed to extract averaged quantities. For solid systems, of particular interest are quantities that
correspond to variables in elasticity models, e.g., displacement, strain and stress. Therefore, it is of great
practical interest to develop appropriate formulas that map particle trajectories to continuum variables
while still following fundamental principles in continuum mechanics.
For homogeneous systems, the stress can be obtained from the second law of thermodynamics
\cite{Clausius1870,Maxwell1870,maxwell1874,tsai1979,mclellan1974,rowlinson2003molecular,swenson1983,ziesche1988}. The stress obtained this way is known as the virial stress. For non-homogeneous systems,
the local stress (in space and time) can be defined following the seminal work of Irving and Kirkwood \cite{irving1950}, where point distributions were introduced to represent a continuous displacement and velocity field. The formula for the stress is derived based on conservation laws, which is consistent with
the continuum mechanics framework. A more practical approach was proposed by Hardy \cite{hardy1982}, where a smooth kernel function is used to replace the point distributions.
This approach leads to a spatially averaged stress,  called the {\it Hardy stress}, derived also from conservation laws. Similar work can be found in \cite{MuBe93,MuBe94,murdoch2005,Costanzo2005533,costanzo2004notion,Chen2003359,Chen2003377}. Recently, Yang et al \cite{yang2012} proposed to further generalized the formulas to involve the temporal and spatial averages simultaneously.

Zimmerman \cite{zimmerman2004calculation} discussed several practical issues with the calculation of the Hardy stress, including the choice of kernel functions and the sample radius.  However, to our knowledge, there has been no quantitative estimate of the sampling error.   In this paper, we analyze the spatial sampling error at zero temperature
and relate the error to certain moment conditions for the kernel function. Further, we suggest a simple
procedure to construct hybrid kernel functions to satisfy the moment conditions and to significantly improve
the sampling accuracy.

To study the accuracy within a standard numerical analysis framework, we have in mind a smooth function, with values only given at atomic positions. Our goal is then to reconstruct the smooth function with maximal order of accuracy. We employ a kernel function to form a set of basis functions  centered at each atom, and
use the values of the function as the nodal values. Using Taylor expansions, we have obtained the leading terms in the error. By choosing kernel functions for which the leading terms are small or zero, we maximize the order of accuracy. We will show that these kernel functions also lead to more accurate formulas for the stress calculation.

The rest of the paper is organized as follows. In section \ref{sec: derivation}, we show a simple derivation of the Hardy stress for molecular mechanics models. In section 3, we provide error estimate for the displacement, deformation gradient and the stress. Then in section 4, a procedure for constructing accurate kernel functions is presented. We show some numerical tests in section 5.

\section{The Irving-Kirkwood Type of Formulas}\label{sec: derivation}
In this section, we briefly describe how continuum quantities can be defined from the atomic positions. We will focus primarily on the formulas proposed by Hardy \cite{hardy1982}.  These formulas  can be derived from a molecular dynamics model based on the conservation of momentum and energy \cite{hardy1982}.  Here, we present a similar derivation starting with a molecular mechanics model,  given by
\begin{equation}\label{eq: ms}
  - \nabla_{\mb x_i}V + \mb b_i=0,
\end{equation}
where \(\mb b_i\) is a body force. We will let \(\mb f_i= - \nabla_{\mb x_i}V\) be
the inter-molecular force on atom \(i\), and we assume that \(\mb f_i\) admits the following decomposition,
\begin{equation}\label{eq: f-sum}
  \mb{f}_i= \sum_{j\ne i} \mb{f}_{ij},\quad \mb{f}_{ij}=-\mb{f}_{ji}.
\end{equation}

This leads to
\begin{equation}\label{eq: ms'}
   -\sum_{j\ne i} \mb f_{ij} = \mb b_i.
\end{equation}
Now we define a continuously distributed force,
\begin{equation}\label{eq: ik-b}
  \widetilde{\mb b}(\mb x)= \sum_j \mb b_j \varphi(\mb x-\mb x_j).
\end{equation}

This formula provides a means to sample the body force at any point in space. The kernel function is typically assumed to be symmetric with compact support, i.e.,
\begin{equation}\label{eq: symm}
 \varphi(\bm x)  = \varphi(-\bm x),
\end{equation}
and
\begin{equation}\label{eq: int=1}
 \int \varphi(\bm x) d\bm x =1.
\end{equation}
In practice, we may start with a non-dimensional function $\varphi_0(\bm x)$, and $\varphi_0(\bm x)=0$
when $|\bm x| >1$. Then we choose a sample radius $R_c$, and let \[
\varphi(\bm x) = \frac{1}{{R_c}^3} \varphi_0(\frac{\bm x}{R_c}).\]

\medskip

Multiplying the equation (\ref{eq: ms'}) by \(\varphi(\mb x-\mb x_j)\), we find that,
\begin{equation}\label{eq: elast}
 - \nabla \cdot \widetilde{ \bsigma} = \mb b(\mb x),
\end{equation}
where the expression \(\widetilde{\bsigma}\) is given by,
\begin{equation}\label{eq: hardy}
\widetilde{\bsigma}(\mb x)= -\frac{1}{2}\sum_i \sum_{j\neq i}\mb{f}_{ij} \otimes \mb{x}_{ij} b_{ij}(\mb x),
\end{equation}
and,
\begin{equation}\label{eq: bij}
  b_{ij}(\mb{x})= \dsp\int_0^1 \varphi\left(\mb{x}-(\mb{x}_i-\lambda \mb{x}_{ij})\right)d\lambda.
\end{equation}
Here $\mb{x}_{ij}=\mb{x}_{i}-\mb{x}_{j}.$
The tensor  \(\widetilde{\bsigma}\)  is identified as the first Piola-Kirchhoff stress since equation (\ref{eq: elast})
has the same form as the continuum elastostatics model. The atomic expression of the stress (\ref{eq: hardy})  is in  the same form as the Hardy stress derived from molecular dynamics models \cite{hardy1982,root2003,hardy2002}.

However, this formulation only involves body forces and stress. The formula for the
displacement is not part of this formulation. One may use a formula similar to (\ref{eq: ik-b}) to
define the local displacement. However, it turns out that the formula is inconsistent with the continuum
elastodynamics model. In this case,  we have to rely on the Irving-Kirkwood formalism,  where the {\it local} velocity is given by,
\begin{equation}\label{eq: ik-v}
  \widetilde{\mb v}(\mb x,t)= \frac{1}{\widetilde{\rho}} \sum_i m_i \mb v_i(t) \varphi(\mb x-\mb x_i),
\end{equation}
in which,
\begin{equation}\label{eq: rho}
  {\widetilde{\rho}}(\mb x)= \sum_i m_i \varphi(\mb x-\mb x_i).
\end{equation}

In order to be consistent with the definition of local velocity, we define
the displacement as follows,
\begin{equation}\label{eq: exp-disp}
  \widetilde{\mb u}(\mb x,t)= \frac{1}{\widetilde{\rho}} \sum_i m_i \mb u_i(t) \varphi(\mb x- \mb x_i).
\end{equation}
As a result, we have,
\(  \frac{\partial}{\partial t} \widetilde{\mb u}=  \widetilde{\mb v}.\)
In addition, by taking the time derivative of (\ref{eq: ik-v}), we obtain an equation that resembles the elastodynamics equation,
\begin{equation}
  \widetilde{\rho}\frac{\partial^2}{\partial t^2} \widetilde{\mb u} = \nabla \cdot \widetilde{\bsigma}.
\end{equation}
Here $\widetilde{\bsigma}$ is also given by (\ref{eq: elast}).

The displacement and the deformation gradient are given respectively by,
\begin{equation}\label{eq: I-K-eps'}
\left\{
\eqalign{
\widetilde{\mb u}(\mb x) &= \frac{  \widetilde{\mb q}(\mb x)}{\widetilde{\rho}(\mb x)}, \cr
\nabla \widetilde{\mb u}(\mb x)&=\frac{\nabla \widetilde{\mb q}}{\tilde{\rho}} - \frac{\widetilde{\mb q}\otimes \nabla \widetilde{\rho}}{\widetilde{\rho}^2},
}
\right.
\end{equation}

where
\begin{equation}\label{eq: q}
  \widetilde{\mb q}(\mb x)= \sum_i m_i \mb u_i \varphi(\mb x- \mb x_i).
\end{equation}

In this work, we are interested in the case when the system is made of a single crystal, for which \(m_i=m\),
and the underlying lattice is a simple lattice. This is the case where the accuracy  can be significantly improved.

\section{Error Estimate}

We now turn to the error estimate of the formulas presented in the previous section. These estimates
will be summarized as several theorems. These theorems are presented not for the purpose of mathematical
rigor, but for the purpose of clarity.

\subsection{Approximation of a function at the lattice points}

We notice that both of the formulas (\ref{eq: rho}) and (\ref{eq: q}) are in the following form,
\begin{equation}\label{eq: w}
\rho_0\widetilde{\mb w}(\mb x)=\sum_j \mb w_j\varphi(\mb x-\mb x_j).
\end{equation}
Here the function $\varphi$ can be considered as a regularization of the delta function to smoothly
sample data over a set of discrete points. \(\rho_0\) is
the density in the reference lattice, \(\rho_0= \Omega_0^{-1}\), with \(\Omega_0\) being
the volume of the unit cell.
We will always assume that \(\varphi\) is $C^1$ with compact support, and it satisfies (\ref{eq: symm})
and (\ref{eq: int=1}).

To formulate the problem as an interpolation problem in standard numerical analysis, we
suppose that $\mb w_i$ corresponds to a smooth function, for instance, $\mb w_i=\mb w(\mb x_i)$, where \(\mb w(\mb x)\), typically unknown, has enough smoothness.
In addition, the lattice spacing, denoted by $\epsilon$, is regarded as a small parameter.  It means that over the atomic scale, $\bm w$ is a slowly varying function. In this case, we wish to be able to reproduce $\mb w(\mb x_i)$. Namely \(\widetilde{\mb w}(\mb x)\approx \mb w(\mb x)\) with high accuracy. To this end, we first expand \(\mb w(\bm x)\) in Taylor series,
\begin{equation}\label{eq: taylor}
\eqalign{
{\mb w}(\mb x_j)&=\mb w(\mb x_i)+(\mb x_j-\mb x_i)\cdot\nabla \mb w(\mb x_i)\\
&+\frac{1}{2}((\mb x_j-\mb x_i)\cdot\nabla)^2 \bm w(\mb x_i)+O(\epsilon^3)
}
\end{equation}
Combining the equations (\ref{eq: w}) and (\ref{eq: taylor}) and collecting terms,  we found
\begin{equation}
\eqalign{
\rho_0\widetilde{\mb w}(\mb x_i)&=\sum_j\varphi(\mb x_i-\mb x_j)\mb w(\mb x_i)\\
&+\frac{1}{2}\sum_j\varphi(\mb x_i-\mb x_j)\bigl((\mb x_i-\mb x_j)\cdot\nabla \bigr)^2 \mb w(\mb x_i)+O(\epsilon^4).
}
\end{equation}

In such an expansion, we can eliminate terms of the following form,
\[ \sum_j \varphi(\mb x_i-\mb x_j) (\mb x_i - \mb x_j)^n,\]
for odd powers of \(n\). This is due to the symmetry of the lattice and the kernel function \(\varphi\).

Using the translational symmetry, we may define the {\it moments},
\begin{equation}
m_0=\sum_j\varphi(\mb x_j),\quad m_2^{\alpha,\beta}=\sum_j \mb x_j^\alpha \mb x_j^\beta \varphi(\mb x_j).
\end{equation}
Here $\mb x_j^\alpha$ denotes a component of the coordinate of the $j$th atom.
As a result we have,
\begin{equation}\label{eq: disp_taylor}
\rho_0\widetilde{\mb w}(\mb x_i)= m_0 \mb w(\mb x_i)+\frac{1}{2} m_2:\nabla^2 \mb w(\mb x_i)+O(\epsilon^4).
\end{equation}

This leads to the following error estimate.
\begin{theorem}
  Suppose that \(\mb w\in C^2\), and
  \begin{equation}\label{eq: mc1}
    m_0= \rho_0,
 \end{equation}
then,
\begin{equation}
  \widetilde{\mb w}(\mb x_i) = \mb w(\mb x_i) + \mathcal{O}(\epsilon^2).
\end{equation}
If we further assume that \(\mb w\in C^4\) and that
\begin{equation}\label{eq: mc2}
m_2= 0,
\end{equation}
then,
\begin{equation}
  \widetilde{\mb w}(\mb x_i) = \mb w(\mb x_i) + \mathcal{O}(\epsilon^4).
\end{equation}
\end{theorem}

The conditions (\ref{eq: mc1}) and (\ref{eq: mc2}) will be referred to as {\it moment conditions}.
Although the first moment condition (\ref{eq: mc1}) can be viewed as a quadrature approximation of the equation
(\ref{eq: int=1}), in general (\ref{eq: int=1}) does not imply that the first moment condition is exactly satisfied. This has been observed in \cite{fu2013on, fu2013modification}.
In practice, when the sample radius $R_c$ is small, $m_0$ may not be close to $\rho_0$, and the lemma
indicates that large error may occur.

\subsection{Approximation of the gradient}

Taking the gradient of (\ref{eq: w}), we find that,
\begin{equation}\label{eq: w-grad}
\rho_0\nabla \widetilde{\mb w}(\mb x)=\sum_j \mb w_j \nabla\varphi(\mb x-\mb x_j).
\end{equation}
Inserting (\ref{eq: taylor}) into the equation above, we find that,
\begin{equation}
\eqalign{
\rho_0\nabla \widetilde{\mb w}(\mb x_i)&= {\mb w}(\mb x_i)\sum_j \nabla \varphi(\mb x_i-\mb x_j)\\
&+\sum_j \nabla \varphi(\mb x_i-\mb x_j) \otimes (\mb x_j-\mb x_i) \nabla \mb w(\mb x_i) + \mathcal{O}(\epsilon^3).
}
\end{equation}
We can eliminate the first term  using
the antisymmetry of $\nabla \varphi$. Using the translational symmetry, we define another moment,
\begin{equation}\label{eq: mu1}
  \mu_1^{\alpha, \beta}= -\sum_j\partial_{\mb x^\alpha} \varphi(\mb x_j) \mb x_j^\beta.
\end{equation}
This leads to:
\begin{lemma}\label{lem}
Suppose that \(\mb w\in C^3(\mathbb{R})\),  and
\begin{equation}\label{eq: mc3}
\mu_1= \rho_0 {\rm I},
\end{equation}
where $\rm I$ is the identity matrix, then
\begin{equation}
  \nabla \widetilde{\mb w}(\mb x_i) = \nabla \mb w(\mb x_i) + \mathcal{O}(\epsilon^3).
\end{equation}
\end{lemma}

Enforcing (\ref{eq: mc3}) in addition to the moment conditions (\ref{eq: mc1}) and (\ref{eq: mc2})
would certainly impose too many constraints on the kernel function. Notice that $\mu_1$ corresponds to the integral $\int \nabla \varphi(\bm x) \otimes \bm x d\bm x$. In fact, the condition (\ref{eq: mc3}) is a discrete analog of the identity,
\begin{equation}
\int \nabla \varphi(\bm x) \otimes \bm x d\bm x = - \rm I,
\end{equation}
which can be derived from (\ref{eq: int=1}). Therefore  the moment conditions (\ref{eq: mc1}) and (\ref{eq: mc3}) are not independent constraints.

To illustrate the connection between the moment conditions (\ref{eq: mc1}) and (\ref{eq: mc3}), we
express each lattice point in terms of three basis vectors, $\bm b_1$, $\bm b_2$ and $\bm b_3$,
\begin{equation}
\bm x = x_1 \bm b_1 + x_2 \bm b_2 + x_3 \bm b_3.
\end{equation}
Here $x_1= i \epsilon$, $x_2= j \epsilon$, and $x_3= k \epsilon$, with $i, j,$ and $k$ being integers.
We also need the reciprocal basis $\bm \xi_\ell$, with the orthogonality condition,
\begin{equation}
 \bm \xi_m \cdot \bm b_n = \delta_{m,n}, \quad m,n=1,2,3.
\end{equation}

Next we let $\psi(x_1, x_2, x_3)= \varphi(\bm x)$. A direct calculation yields,
\begin{equation}
\nabla \varphi(\bm x)= \partial_{x_1}\psi \bm \xi_1+\partial_{x_2}\psi \bm \xi_2+\partial_{x_3}\psi \bm \xi_3.
\end{equation}
 Therefore, the right hand side of (\ref{eq: mu1}) will contain terms in the form of,
 \[ \sum_i \sum_j \sum_k \partial_{x_m} \psi(x_1, x_2, x_3) x_n.\]
For example, we can estimate,
\begin{equation}
\eqalign{
&\sum_i \sum_j \sum_k \partial_{x_1} \psi(x_1, x_2, x_3) x_1\\
=& \epsilon  \sum_j \sum_k \sum_i \partial_{x_1} \psi(i \epsilon, j \epsilon, k\epsilon)  i \\
=&  \sum_j \sum_k \sum_i \big[\psi((i+1) \epsilon, j\epsilon, k \epsilon) - \psi(i \epsilon, j \epsilon, k\epsilon) \big] i \\
&-  \half  \epsilon^2 \sum_j \sum_k \sum_i \partial_{x_1^2}^2 \psi(i \epsilon, j \epsilon, k\epsilon)  i
 + \mathcal{O}(\epsilon^3).\\
=& - \sum_{\bm x} \varphi(\bm x) + \cal{O}(\epsilon^3).
}
\end{equation}
In the last step, we have applied a summation-by-parts along the $\bm b_1$ direction, and then switched
the summation back to the summation over $\bm x$. Further, the second term has been eliminated using
the symmetry of the lattice and the function $\varphi.$ Following the same procedure, one can show that
\[\sum_i \sum_j \sum_k \partial_{x_1} \psi(x_1, x_2, x_3) x_2 = \cal{O}(\epsilon^3).\]
Combining these results, we find that,
\begin{equation}
\mu_1 = -m_0 \rm I + \cal{O}(\epsilon^3).
\end{equation}
As a result, we have,
\begin{theorem}
If the condition (\ref{eq: mc3}) in Lemma \ref{lem} is replaced by the moment condition (\ref{eq: mc1}),
the same results hold.
\end{theorem}

\subsection{Approximation of the stress}
In this section, we focus on the implication of the moment conditions to the accuracy of the calculated stress. In order to assess the accuracy, we need to define the {\it exact stress}. This is a concept subject to
a great deal of debate \cite{tadmor2010}. However, for a uniformly deformed system, with either infinite size or periodic boundary conditions, the stress is well defined, and it is given by the formula,
\begin{equation}
\bsigma= -\frac{1}{2\Omega_0} \sum_j \bm f_{ij}\otimes \bm x_{ij}.
\end{equation}
Here $\Omega_0$ is the volume of a primitive cell. Since the system is uniform, the stress does not depend on the choice of $i$. We have,

\begin{theorem}
For a system with uniform deformation gradient, one has the error estimate:
\begin{equation}\label{eq: th1}
\widetilde{\bsigma}(\mb x_i)=\bsigma+\mathcal{O}(\epsilon^2),
\end{equation}
under the condition (\ref{eq: mc1}).
\end{theorem}
\begin{proof}
We begin with the expression of the stress (\ref{eq: hardy}) and (\ref{eq: bij}). For a system with uniform deformation gradient, it can be easily verified that the inversion symmetry is still preserved.
In addition, the force $\bm f_{ij}$ only depends on the
relative position of the two atoms. Therefore, it is enough to consider the forces around a particular atom, say, $0$. Then we may rewrite the stress as,
\begin{equation}
\widetilde{\bsigma}(\bm x)= -\half \sum_j \bm f_{0j} \otimes \bm x_{0j} \sum_i \int_0^1 \varphi(\bm x-\bm x_i +\lambda  \bm x_{0j})d\lambda.
\end{equation}
Using the Taylor expansion,
\[
\eqalign{
\varphi(\bm x- \bm x_i + \lambda \bm x_{0j})&= \varphi(\bm x- \bm x_i) + \lambda \bm x_{0j}\cdot
\nabla \varphi(\bm x- \bm x_i) \\
&+ \half  \lambda^2 (\bm x_{0j}\cdot
\nabla \varphi)^2(\bm x- \bm x_i) +
\cdots,
}
\]
we get,
\[
\eqalign{
\widetilde{\bsigma}(\bm x)&= -\half \sum_j \bm f_{0j} \otimes \bm x_{0j} \sum_i \varphi(\bm x- \bm x_i)\\
&-\frac14 \sum_j \bm f_{0j} \otimes \bm x_{0j} \sum_i \bm x_{0j}\cdot
\nabla \varphi(\bm x-\bm x_i) + \mathcal{O}(\epsilon^2).
}
\]
If $\bm x$ is a lattice point, and the  moment condition (\ref{eq: mc1}) is satisfied, then,
\[  -\half \sum_j \bm f_{0j} \otimes \bm x_{0j}\sum_i \varphi(\bm x- \bm x_i) = -\frac{1}{2\Omega_0}
 \sum_j \bm f_{0j} \otimes \bm x_{0j} = \bsigma.\]
The second term on the right hand side becomes zero because of the inversion symmetry. This completes
the proof.
\end{proof}

\section{Constructing Hybrid Kernels with High Order of Accuracy}

Usually, the kernel functions suggested by previous works \cite{hardy1982,hardy2002,yang2012} do not satisfy the moment conditions exactly.
Therefore, the accuracy is not guaranteed. To ensure accuracy, we will construct hybrid kernel functions
from existing ones. Notice that the second moment condition (\ref{eq: mc2}) has 9 equations. However, for cubic crystalline systems, these equations can be significantly reduced,

\begin{theorem}
For cubic systems (b.c.c and f.c.c) the second  moment is $C\rm I$, where $C$ is a constant and $\rm I$ is the $3\times 3$ identity matrix. \end{theorem}
This can be proved by expressing the coordinate of each atom in terms of the basis vectors, and
use the symmetry of the cubic lattice and the kernel function. For brevity, we will skip the proof.

We now describe the procedure to construct hybrid kernel functions. To begin with, let $\varphi_1$ and $\varphi_2$ be two kernel functions that satisfy (\ref{eq: int=1}) and the symmetry (\ref{eq: symm}).
We define,
\begin{equation}
\varphi_{\rm hybrid}=A_1\varphi_1+A_2\varphi_2.
\end{equation}
with the two constants $A_1$ and $A_2$ satisfying,
\begin{equation}\label{eq: hybrid1}
A_1 + A_2 =1,
\end{equation}
in order for $\varphi_{\rm hybrid}$ to satisfy (\ref{eq: int=1}). Further, to satisfy the moment condition (\ref{eq: mc1}) exactly,
we need,
\begin{equation}\label{eq: hybrid2}
A_1 \sum_i \varphi_1(\bm x_i) + A_2 \sum_i \varphi_2(\bm x_i) = \rho_0.
\end{equation}
These two coupled equations can be solved to obtain the two parameters.  A similar procedure can be
applied to incorporate the second moment condition (\ref{eq: mc2}).

For the purpose of this construction, it is important to have a large collection of independent kernel functions to work with. Some examples are provided here. They are already scaled properly to satisfy
(\ref{eq: int=1}).

\begin{enumerate}
\item{\bf Cubic spline function.}
\begin{equation}\label{eq: kern1}
  \varphi_0^{I}(\mb x)=
  \left\{
    \begin{array}{lc}
\dsp      \frac{15}{4\pi}\bigl[1+\left(2{r}-3\right)r^2\bigr],  &r\le 1,\\
      0,  &r>1.
    \end{array}
  \right.
\end{equation}
Here $r = \sqrt{x_1^2+x_2^2+x_3^3}$ is the length of the vector.

\item{\bf A regularized step function \cite{yang2012}.}
\begin{equation}\label{eq: kern2}
  \varphi_0^{II}(\mb x)=
  \left\{
    \begin{array}{lc}
\dsp  \frac{1}{c^{II}}\exp \frac{0.1}{r^2-1}, &r\le 1,\\
      0,  &r>1.
    \end{array}
  \right.
\end{equation}
Here $r=\sqrt{x_1^2+x_2^2+x_3^2}$, \(c^{II}\approx 2.77442\) is a normalization constant.

\item{\bf Cosine kernels \cite{hardy2002}}
\begin{equation}\label{eq: kern3}
  \varphi_0^{III}(\mb x)=
  \left\{
  \begin{array}{lc}
  \frac{1}{8}\prod_{i=1}^{3}\bigl(1+\cos \pi x_i\bigr),  &|x_i|\le 1,i=1,2,3,\\
    0,  &\dsp{otherwise}.
  \end{array}
  \right.
\end{equation}

\item{\bf Gaussian kernel \cite{hardy1982}}
\begin{equation}\label{eq:kern4}
  \varphi_0^{IV}(\mb x)=
  \left\{
  \begin{array}{lc}
  \bigl(\frac{3}{0.997\sqrt{2\pi}}\bigr)^3\bigl(e^{-\frac{9x_1^2+9x_2^2+9x_3^2}{2}}\bigr),  &|x_i|\le 1,i=1,2,3,\\
    0,  &\dsp{otherwise}.
  \end{array}
  \right.
\end{equation}

\item{\bf A truncated polynomial function \cite{root2003}}
\begin{equation}\label{eq:kern5}
  \varphi_0^{V}(\mb x)=
  \left\{
  \begin{array}{lc}
    \bigl(\frac{15}{16}\bigr)^3\prod_{i=1}^{3}\bigl(1-2x_i^2+x_i^4\bigr),  &|x_i|\le 1,i=1,2,3,\\
    0,  &\dsp{otherwise}.
  \end{array}
  \right.
\end{equation}
\end{enumerate}

\section{Numerical Experiments}

We first consider a displacement field given by the following elementary functions, defined
on a f.c.c. lattice, with lattice constant $\epsilon$,
\begin{equation*}
\eqalign{
g_0(\mb x)&=\frac{2\epsilon}{100},\\
g_1(\mb x)&=\frac{2\epsilon}{100}(\frac{2\mb x}{\epsilon}),\\
g_2(\mb x)&=\frac{2\epsilon}{100}(\frac{2\mb x}{\epsilon})^2,\\
g_3(\mb x)&=\frac{2\epsilon}{100}(\frac{2\mb x}{\epsilon})^3,\\
g_4(\mb x)&=\frac{2\epsilon}{100}(\frac{2\mb x}{\epsilon})^4,\\
g_5(\mb x)&=\frac{2\epsilon}{100}\sin(\frac{2\pi \mb x}{\epsilon}).
}
\end{equation*}
For each of these cases, we compute the approximation function $\widetilde{\mb w}(x)$ using
all the five kernel functions. The cut-off distance of the kernel function is set to be $R_c=4\epsilon$
and the simulation box is of the size $20\epsilon\times 20\epsilon \times 10\epsilon$. The error, measured in the maximum norm,  is listed in table \ref{tab:psi}.
We observe that starting from $g_2$, the approximation error becomes large. We have chosen $R_c$ large enough so that the first moment condition (\ref{eq: mc1}) is satisfied. Also shown in the table is
the second moment of each kernel. One can see that the moment is away from zero.
    \begin{table}[h]
    \begin{center}
      \caption{Error $\|\mb w(x)-\tilde{\mb w}(x)\|_{\infty}$ without the moment conditions (\ref{eq: mc2}). \label{tab:psi}}
      \begin{tabular}{cccccc}
        \toprule
        $\mb w$ &
        $\|\mb w - \tilde{\mb w}^{I}\|_{\infty}$&
        $\|\mb w - \tilde{\mb w}^{II}\|_{\infty}$&
        $\|\mb w - \tilde{\mb w}^{III}\|_{\infty}$&
        $\|\mb w - \tilde{\mb w}^{IV}\|_{\infty}$&
        $\|\mb w - \tilde{\mb w}^{V}\|_{\infty}$\\
        \midrule
        $g_0$       & 0.000000E+00     & 0.000000E+00    & 0.000000E+00 & 0.000000E+00 & 0.000000E+00\\
        $g_1$       & 0.288384E-11     & 0.352140E-11    & 0.313907E-11 & 0.288384E-11 & 0.243145E-11\\
        $g_2$     & 0.192054E+04     & 0.191935E+04    & 0.186910E+04 & 0.192054E+04 & 0.188184E+04\\
        $g_3$     & 0.134305E+06     & 0.134221E+06    & 0.130707E+06 & 0.134305E+06 & 0.131599E+06\\
        $g_3$     & 0.842426E+07     & 0.841901E+07    & 0.819845E+07 & 0.842426E+07 & 0.825439E+07\\
        $g_4$ & 0.265355E+03     & 0.265257E+03    & 0.261305E+03 & 0.265355E+03 & 0.262241E+03\\
        $m_2$     & 0.605382E+01     & 0.604951E+01    & 0.586577E+01 & 0.605382E+01 & 0.591210E+01\\
        \bottomrule
      \end{tabular}
      \end{center}
    \end{table}

We now construct a hybrid kernel function from   \(\varphi^{I}\)
and   \(\varphi^{II}\).  For an iron-alpha system with b.c.c. structure
and lattice spacing $\epsilon=2.865($\AA$)$, we have listed the results in table \ref{tab:u_error_bcc}.
For an aluminum system with f.c.c structure and lattice spacing $\epsilon=4.032($\AA$)$, the results are listed in table \ref{tab:u_error_fcc}. Clearly, from the first two rows, we can see that the error is nearly zero when the first moment condition (\ref{eq: mc1}) is satisfied,  confirming that each kernel function will exactly reproduce constant and linear functions.  From the third and forth rows in tables \ref{tab:psi}, \ref{tab:u_error_bcc} and \ref{tab:u_error_fcc}, we see that with the second moment condition (\ref{eq: mc2}), we are also able to approximate accurately quadratic and cubic functions.

  \begin{table}[htbp]
  \begin{center}
  \caption{Displacement error for a b.c.c. system.} \label{tab:u_error_bcc}
  \begin{tabular}{cccc}
    \toprule
    $\mb w$&
    $\|\mb w-\tilde{\mb w}^{Hybrid} \|_{\infty}$&
    $\|\mb w - \tilde{\mb w}^{I}\|_{\infty}$&
    $\|\mb w - \tilde{\mb w}^{III}\|_{\infty}$
    \\
    \midrule
    $g_0$       &    0.10824674E-14 & 0.11102230E-15 & 0.27755576E-16\\
    $g_1$       &    0.52735594E-15 & 0.16653345E-15 & 0.13877788E-15\\
    $g_2$       &    0.47184479E-15 & 0.13825538E-01 & 0.16859664E-01\\
    $g_3$       &    0.47184479E-15 & 0.13825538E-01 & 0.16859664E-01\\
    $g_4$       &    0.23260274E-02 & 0.85314910E-01 & 0.10454843E+00\\
    $g_5$       &    0.56716289E-03 & 0.16469294E-01 & 0.19959146E-01\\
    \bottomrule
  \end{tabular}
  \end{center}
  \end{table}

  \begin{table}[htbp]
  \begin{center}
  \caption{Displacement error for a f.c.c. system. \label{tab:u_error_fcc} }
  \begin{tabular}{cccc}
    \toprule

    $\mb w$&
    $\|\mb w-\tilde{\mb w}^{Hybrid} \|_{\infty}$&
    $\|\mb w - \tilde{\mb w}^{I}\|$&
    $\|\mb w - \tilde{\mb w}^{III}\|$
    \\
    \midrule
    $g_0$       &    0.15265567E-14 & 0.99920072E-15 & 0.10269563E-14\\
    $g_1$       &    0.74940054E-15 & 0.11379786E-14 & 0.49960036E-15\\
    $g_2$       &    0.63837824E-15 & 0.13831333E-01 & 0.16859664E-01\\
    $g_3$       &    0.47184479E-15 & 0.41493998E-01 & 0.50578993E-01\\
    $g_4$       &    0.23340091E-02 & 0.85350211E-01 & 0.10454843E+00\\
    $g_5$       &    0.56909837E-03 & 0.16476307E-01 & 0.19959146E-01\\
    \bottomrule
  \end{tabular}
  \end{center}
  \end{table}

\medskip

In the next numerical experiment, we focus only on the first moment condition (\ref{eq: mc1}). We construct hybrid kernel functions mentioned above.
We solve the linear system composed of (\ref{eq: hybrid1}) and (\ref{eq: hybrid2}). In the case  when the moments of $\varphi_1$ and  $\varphi_2$ are close, a combination may not satisfy the condition (\ref{eq: mc1}). In this case, we choose the hybrid kernel to be the one
whose moment is closer to $\rho_0$, see Figure \ref{fig:hybrid_mc0}. Using this hybrid kernel, we compute the stress of a b.c.c Fe system under uniform 1\% stretch. The results are displayed in Figure \ref{fig:hybrid_stress}.
Clearly the hybrid kernel gives much more accurate results, especially when the
sample radius $R_c>6.5\AA$. More importantly, the error in this regime is quite uniformly controlled, while the error from a single kernel function suffers from significant fluctuations.

\begin{figure}
  \centering
  \includegraphics[scale=0.45,angle=-90]{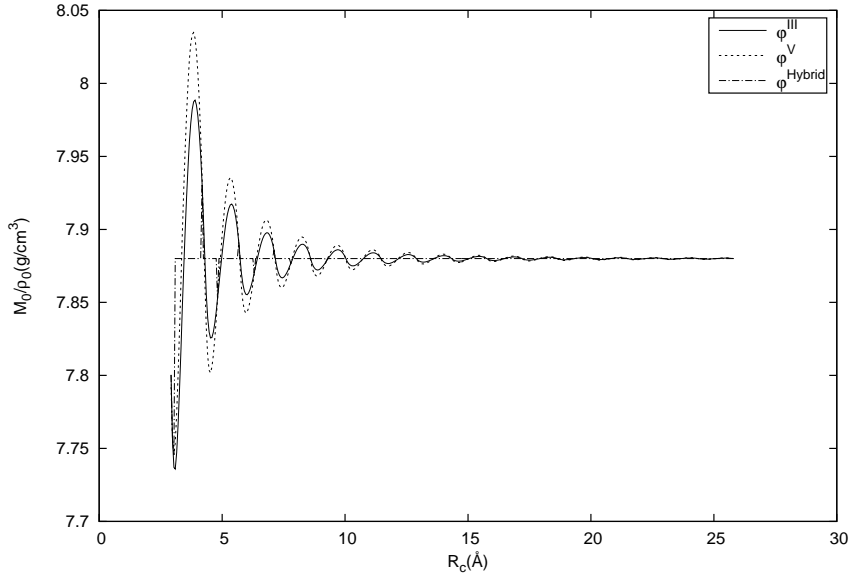}\\
  \caption{The first moment of $\varphi^{III}$ , $\varphi^{V}$ and
  the hybrid kernel based on $\varphi^{III}$ and $\varphi^{V}$ versus the sampling radius ($R_c$).}\label{fig:hybrid_mc0}
\end{figure}

\begin{figure}
  \centering
  \includegraphics[scale=0.45,angle=-90]{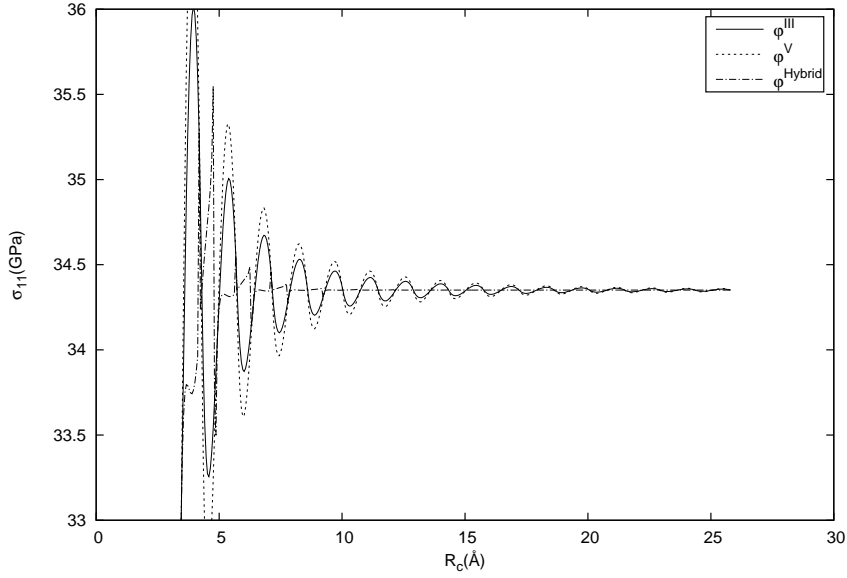}\\
  \caption{Stress for $\varphi^{III}$ , $\varphi^{V}$ and the hybrid kernel based on  $\varphi^{III}$ and $\varphi^{V}$ versus the sampling radius ($R_c$) with 1\% stretch in b.c.c Fe system}\label{fig:hybrid_stress}
\end{figure}

\medskip

To study the accuracy of the computed strain, we computed the approximate gradient and made comparisons with the exact values.
The results for the b.c.c system is shown in table \ref{tab:du_error_bcc}, and for the
f.c.c. system, the results are summarized in table \ref{tab:du_error_fcc}. In both cases,
the hybrid kernel is constructed from the first and third kernels, we see that the
hybrid approach has significantly reduced the approximation error.

  \begin{table}[htbp]
  \begin{center}
  \caption{Deformation gradient error for a b.c.c system.  \label{tab:du_error_bcc}}
  \begin{tabular}{cccc}
    \toprule
    $\mb w$&
    $\|\nabla \mb w-\nabla \widetilde{\mb w}^{Hybrid} \|_{\infty}$&
    $\|\nabla \mb w-\nabla \widetilde{\mb w}^{I}\|_{\infty}$&
    $\|\nabla \mb w-\nabla \widetilde{\mb w}^{III}_{\infty}\|$
    \\
    \midrule
    $g_0$       &    0.40661405E-16 & 0.28940627E-17 & 0.17948354E-17\\
    $g_1$       &    0.44929338E-15 & 0.49453648E-02 & 0.50644210E-02\\
    $g_2$       &    0.92287289E-15 & 0.98907297E-02 & 0.10128842E-01\\
    $g_3$       &    0.59230612E-02 & 0.13790267E-01 & 0.13979664E-01\\
    $g_4$       &    0.23692245E-01 & 0.15598148E-01 & 0.15403288E-01\\
    $g_5$       &    0.84147098E-02 & 0.84147098E-02 & 0.84147098E-02\\
    \bottomrule
  \end{tabular}
  \end{center}
  \end{table}

  \begin{table}[htbp]
  \begin{center}
  \caption{Error of the deformation gradient for a f.c.c system. \label{tab:du_error_fcc}  }
  \begin{tabular}{cccc}
    \toprule
    $\mb w$&
    $\| \nabla {\mb w}-\nabla \widetilde{\mb w}^{Hybrid} \|_{\infty}$&
    $\| \nabla {\mb w}-\nabla \widetilde{\mb w}^{I}\|_{\infty}$&
    $\| \nabla {\mb w}-\nabla \widetilde{\mb w}^{III}\|_{\infty}$
    \\
    \midrule
    $g_0$       &    0.27415087E-16 & 0.47228050E-17 & 0.24982699E-17\\
    $g_1$       &    0.36429193E-16 & 0.64767543E-05 & 0.12884199E-03\\
    $g_2$       &    0.34694470E-16 & 0.12953509E-04 & 0.25768398E-03\\
    $g_3$       &    0.20740369E-02 & 0.20757141E-02 & 0.20406720E-02\\
    $g_4$       &    0.82961474E-02 & 0.82510422E-02 & 0.91934239E-02\\
    $g_5$       &    0.84147098E-02 & 0.84147098E-02 & 0.84147098E-02\\
    \bottomrule
  \end{tabular}
  \end{center}
  \end{table}

\medskip

  Finally, we consider a crack in a b.c.c. crystal of iron-alpha. The system studied consists of a $3D$ rectangular sample, with three orthogonal axis being along $[110],[1\bar{1}0]$ and $[001]$ directions respectively. The crack is oriented along the $[001][110]$ directions.  The crack is initialized using the anisotropic linear elasticity solution \cite{Sih19671}.  This displacement field from the analytical expressions will be
  regarded as the exact solution. The interatomic potential is the embedded atom potential \cite{ShFa96}.
Based on the displacement of the atoms, we have computed the averaged displacement at each atomic position, and the error is listed in the table \ref{tab:crack_bcc}. Notice that in general, the averaged displacement given by the Hardy's formulas does not agree with the
actual displacement at that point, as discussed in section 3. The error is large when the displacement field is not smooth, as confirmed by Figure \ref{fig: crack-u}, in which large error is observed near the crack-tip. In contrast, the hybrid kernel function yields
much more accurate results.
Finally, in Figure \ref{fig: crack-s}, we show the stress computed using the original kernels and the hybrid kernel.

  \begin{table}[htbp]
  \begin{center}
    \caption{The comparison among $\varphi^{I}$, $\varphi^{III}$ and the hybrid kernel,  \\
 measured by different norms around a crack in a b.c.c. iron system.
    \label{tab:crack_bcc}}
    \begin{tabular}{cccc}
      \toprule
      Norm&
      $\varphi^{I}$&
      $\varphi^{III}$&
      Hybrid $\varphi^{I}$ and $\varphi^{III}$\\
      \midrule
      $\|\cdot\|_{1}$    & 2.306150E+00& 2.806700E+00& 2.640000E-02\\
      $\|\cdot\|_{2}$    & 5.257212E-02& 6.393728E-02& 2.032535E-03\\
      $\|\cdot\|_{\infty}$& 2.840000E-03& 3.380000E-03& 3.700000E-04\\
      \bottomrule
    \end{tabular}
    \end{center}
  \end{table}

  \begin{figure}\label{fig: crack-u}
  \begin{center}
    \includegraphics[scale=0.6]{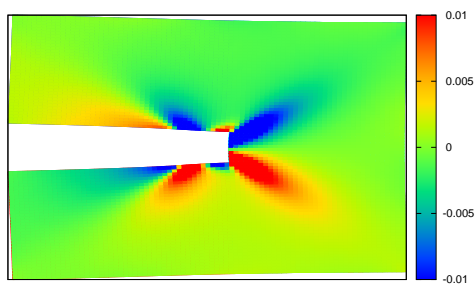}\\
    \includegraphics[scale=0.6]{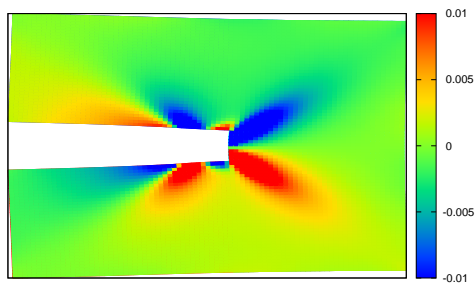}\\
    \includegraphics[scale=0.6]{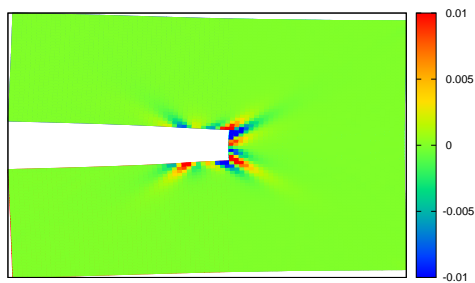}
    \caption{The error of the displacement ($u_2$) for the crack problem. From top to bottom: results computed from $\varphi_1$, $\varphi_2$, and the hybrid kernel.}
  \end{center}
  \end{figure}

  \begin{figure}\label{fig: crack-s}
  \begin{center}
    \includegraphics[scale=0.6]{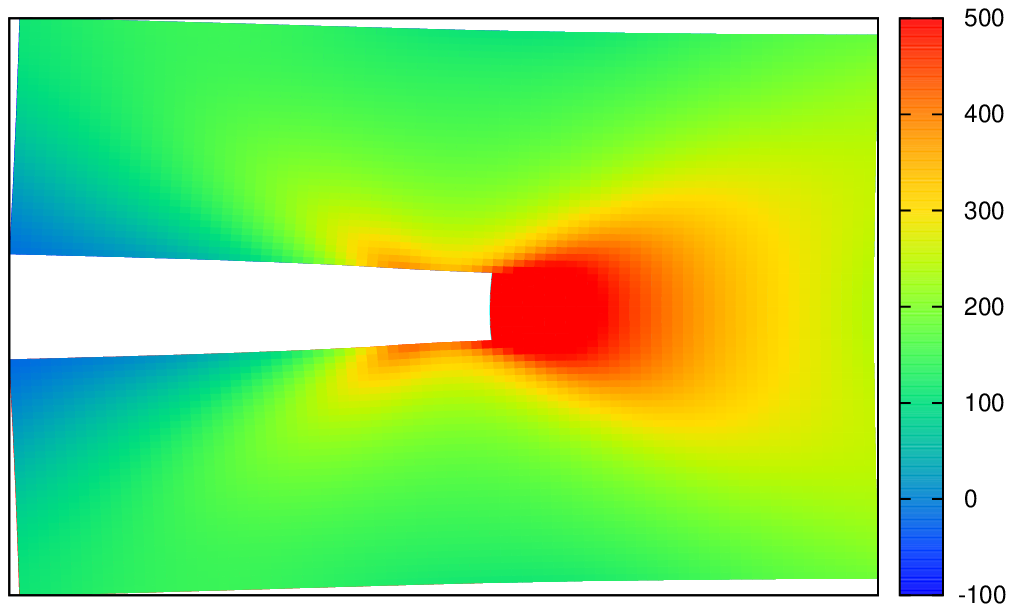}\\
    \includegraphics[scale=0.6]{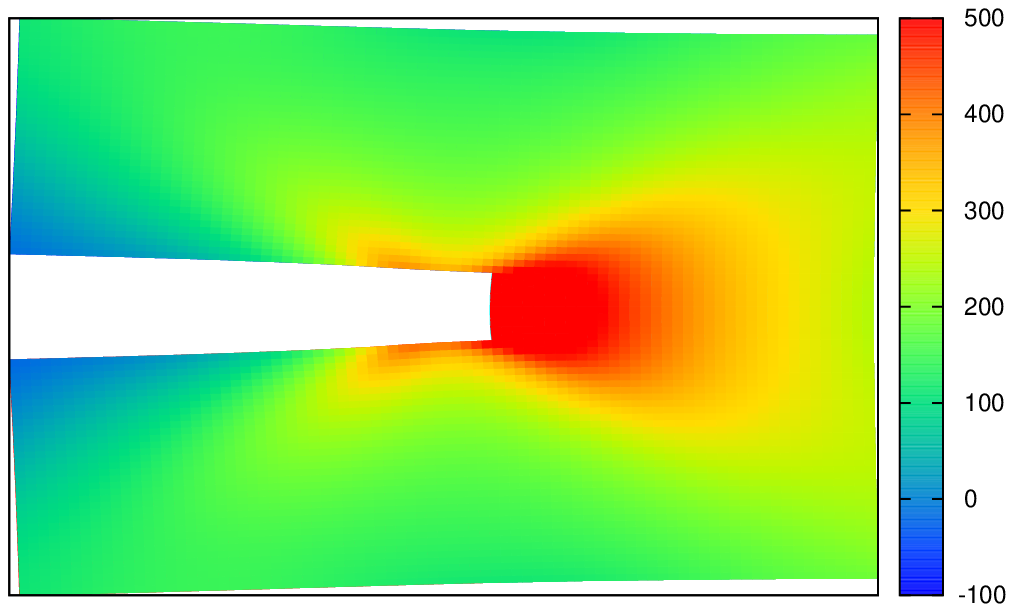}\\
    \includegraphics[scale=0.6]{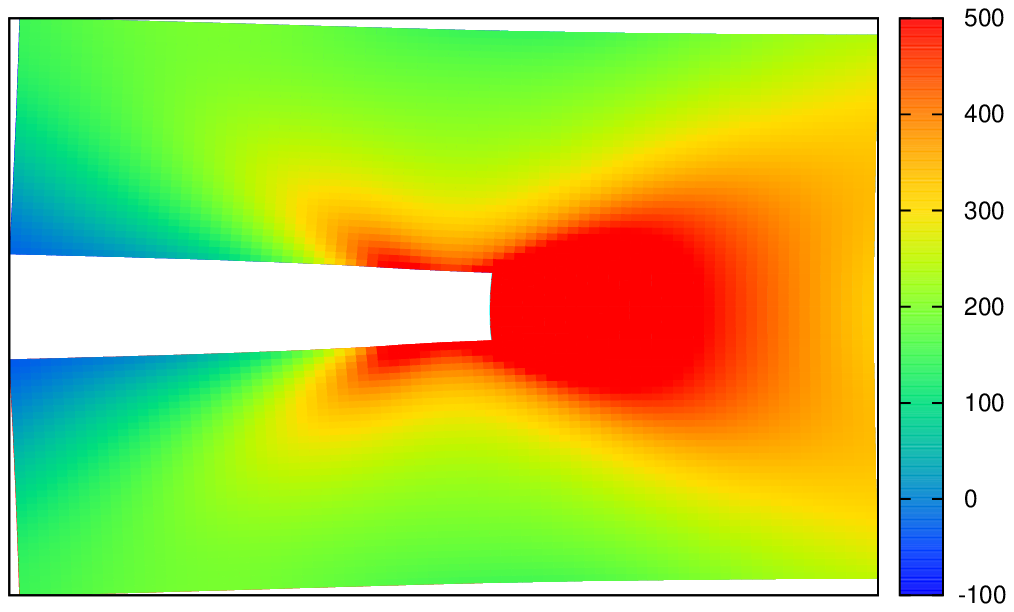}

    \caption{The sampled stress $\sigma_{11}$ in the crack problem. From top to bottom: $\varphi_1$, $\varphi_2$, and the hybrid kernel.}
    \end{center}
  \end{figure}

\section*{Acknowledgments}
The work of Yang was supported by NFSC under the grants 11001210 and 91230203.
The work of Li was initiated and partially fulfilled when
visiting Peking University and Beijing International Center for Mathematical Research. He would like to
express his thanks for the warm and stimulating environment they provided him during his visit.

\end{document}